\documentclass[a4paper,USenglish]{lipics-v2021}
\nolinenumbers
\hideLIPIcs

\usepackage{amsmath,amsfonts,amssymb,mathtools}
\usepackage{makecell}
\usepackage{algorithmicx}
\usepackage[noend]{algpseudocode}
\usepackage{graphicx}
\usepackage{xspace}



\acknowledgements{
We thank the anonymous reviewers of the paper for their helpful comments and suggestions.
We also thank Faith Ellen, Tian Ze Chen and Alexander Spiegelman for the many helpful discussions.}

\funding{This work was supported by the Natural Sciences and Engineering Research Council of Canada (NSERC) through a Discovery Grant and an Undergraduate Student Research Award (USRA).}

\newtheorem{invariant}[theorem]{Invariant}

\providecommand{\ReadOp}[1]{\ensuremath{\mathsf{Read}_{#1}}\xspace}
\providecommand{\Write}{\ensuremath{\mathsf{Write}}\xspace}

\providecommand{\ReadReq}[1]{\ensuremath{\mathsf{ReadReq}_{#1}}}
\providecommand{\ReadAck}[1]{\ensuremath{\mathsf{ReadAck}_{#1}}}
\providecommand{\AttemptReq}[1]{\ensuremath{\mathsf{AttemptReq}_{#1}}}
\providecommand{\AttemptAck}[1]{\ensuremath{\mathsf{AttemptAck}_{#1}}}
\providecommand{\PieceAck}[1]{\ensuremath{\mathsf{PieceAck}_{#1}}}
\providecommand{\PieceReady}[1]{\ensuremath{\mathsf{PieceReady}_{#1}}}
\providecommand{\Mailbox}[1]{\ensuremath{\mathsf{Mailbox}_{#1}}}

\providecommand{\localVal}[1]{\ensuremath{\mathsf{localVal}_{#1}}}
\providecommand{\idx}[1]{\ensuremath{\mathsf{idx}_{#1}}}

\title{Upper and Lower Bounds on the Space Complexity of Multi-word Single-Writer Registers}

\author{Yuanhao Wei}
{University of British Columbia, Vancouver, Canada}
{yuanhaow@cs.ubc.ca}
{https://orcid.org/0000-0002-5176-0961}
{}

\author{Yousof Yavari}
{University of British Columbia, Vancouver, Canada}
{yousofy@student.ubc.ca}
{https://orcid.org/0009-0003-5069-5437}
{}

\authorrunning{Y. Wei and Y. Yavari}

\Copyright{Yuanhao Wei and Yousof Yavari}

\titlerunning{Space Complexity of Multi-word Single-Writer Registers}

\EventEditors{Ioannis Chatzigiannakis, Andrea Vitaletti, Keren Censor-Hillel, and William K. Moses Jr.}
\EventNoEds{4}
\EventLongTitle{40th International Symposium on Distributed Computing (DISC 2026)}
\EventShortTitle{DISC 2026}
\EventAcronym{DISC}
\EventYear{2026}
\EventDate{November 9--13, 2026}
\EventLocation{Rome, Italy}
\EventLogo{}
\SeriesVolume{397}
\ArticleNo{49}

\begin{document}

\maketitle
	
\begin{abstract}    
	We prove matching upper and lower bounds on the space complexity of simulating a large shared register using smaller shared registers. We focus on the case where both the simulated and base registers are single-writer, which means they can be accessed concurrently by multiple readers but only by a single writer. 
    To strengthen our lower bounds, we prove that they hold even when the base registers are atomic and the simulated register is regular. Furthermore, the lower bounds hold for obstruction-free implementations, which means they also hold for lock-free and wait-free implementations. 
    
	If $m$ is the number of values representable by the large register and $b$ is the number of values representable by each base register, our first lower bound says that any obstruction-free implementation that has an invisible reader requires at least $\lceil \frac{m-1}{b-1} \rceil$ base registers. A reader is considered invisible if it never writes to base registers. This lower bound is asymptotically tight for the invisible-reader case and represents an exponential improvement over the previous best known lower bound.
    For the general case, which allows any combination of visible and invisible readers, we prove a \hbox{$\lceil \min(\frac{m-1}{b-1}, r+\frac{\log{m}}{\log{b}}) \rceil$} space lower bound, where $r$ is the number of readers. 

    To show that this lower bound is asymptotically tight, we develop a wait-free algorithm for simulating a multi-word atomic register from atomic base registers using $\Theta(r + \frac{\log{m}}{\log{b}})$ space. Combining this algorithm with known invisible-reader constructions gives a $\Theta(\min(\frac{m}{b}, r + \frac{\log{m}}{\log{b}}))$ space upper bound.
    This improves upon the previously known space upper bound of $\Theta(\min(\frac{m}{b}, r\frac{\log{m}}{\log{b}}))$.

	\end{abstract}
	
	\section{Introduction}
\label{sec:intro}
Shared registers are among the most fundamental abstractions in concurrent and distributed computing. A register supports two operations: a write, which stores a value, and a read, which retrieves one. Registers are classified by their consistency guarantees — safe, regular, and atomic — in increasing order of strength~\cite{lamport1986interprocess}. A safe register guarantees that a read not concurrent with any write returns the last written value, while concurrent reads may return any value in the register's domain. A regular register additionally ensures that a concurrent read returns either the last written value or the value being written concurrently. An atomic register, the strongest of the three, additionally requires that operations appear to take effect instantaneously at some point during their execution, inducing a total order consistent with real time — a property also known as linearizability. In this paper, we focus our attention on regular and atomic registers.

In practice, the values a register must store can be larger than the base registers available in hardware. 
Implementing multi-word shared registers is challenging because data being accessed by slow readers can be overwritten by fast writers.
To ensure the regular and atomic consistency guarantees, careful synchronization and additional base registers are required.
This gives rise to a natural problem in the theory of shared memory algorithms: 
how many $b$-valued regular or atomic registers are needed to implement an $m$-valued register of the same type?


Since Lamport's foundational work~\cite{lamport1986interprocess}, a rich line of research~\cite{Peterson, chaudhuri1994bounds, vidyasankar1988converting, chaudhuri2000one, chen2017step} has studied this question in the setting where the implemented and the base registers are single-writer (SW), meaning that each register can be accessed concurrently by multiple readers but only a single process can write to the register.
This line of work has focused on developing \emph{wait-free} implementations where each read and write is guaranteed to finish within a finite number of its own steps regardless of the schedule.
The step and space complexity of these implementations are summarized in Table~\ref{complexity-table}, where $r$ represents the number of readers.

\begin{table}[t]
\centering
\footnotesize
    \setlength{\tabcolsep}{3pt}
    \resizebox{\textwidth}{!}{\begin{tabular}{ | l | c | c | c | c | c | c | }
    \hline
    \textbf{Algorithm} & \textbf{Simulated} & \textbf{Base} & \textbf{Readers} & \textbf{Space} & \textbf{Read} & \textbf{Write}  \\
    \textbf{} & \textbf{Register} & \textbf{Register} & \textbf{Invisible} & \textbf{} & \textbf{} & \textbf{}  \\
    \hline
     Peterson \cite{Peterson} & atomic & atomic & no & $\Theta(r \frac{\log m}{\log b})$ & $\Theta(\frac{\log m}{\log b})$ & $\Theta(r \frac{\log m}{\log b})$ \\
    \hline
    Chaudhuri and Welch \cite{chaudhuri1994bounds} & regular & regular & yes & $\Theta(\frac{m}{b})$ & $\Theta(\frac{\log m}{\log b})$ & $\Theta(\frac{\log m}{\log b})$  \\
    \hline
    Vidyasankar \cite{vidyasankar1988converting} & atomic & atomic & yes & $\Theta(\frac{m}{\log b})$ & $\Theta(\frac{m}{\log b})$ & $\Theta(\frac{m}{\log b})$  \\
    \hline
    \makecell[l]{Vidyasankar \cite{VIDYASANKAR1991323} + \\ Chaudhuri and Welch \cite{chaudhuri1994bounds}} & atomic & atomic & yes & $\Theta(\frac{m}{b})$ & $\Theta(\frac{\log m}{\log b})$ & $\Theta(\frac{\log m}{\log b})$  \\
    \hline
    Chaudhuri, Kosa and Welch \cite{chaudhuri2000one} & atomic & atomic & yes & $\Theta(m^2)$ & $\Theta(m^2)$ & 1 \\
    \hline
    Chaudhuri, Kosa and Welch \cite{chaudhuri2000one} & regular & regular & yes & $\Theta(m^2)$ & $\Theta(m^2)$ & 1 \\
    \hline
    Chen and Wei \cite{chen2017step} & atomic & atomic & yes & $\Theta(\frac{m^2}{b^2})$ & $\Theta(\frac{\log m}{\log b})$ & $\Theta(\frac{\log m}{\log b})$  \\
    \hline
     Chen and Wei \cite{chen2017step} & atomic & atomic & no & $\Theta(r\frac{\log m}{\log b})$ & $\Theta(\frac{\log m}{\log b})$ & $\Theta(\frac{\log m}{\log b})$ \\
    \hline
    This paper & atomic & atomic & no & $\Theta(r + \frac{\log m}{\log b})$ & $\Theta(\frac{\log^2 m}{\log^2 b})$ & $\Theta(r+ \frac{\log m}{\log b})$ \\
    \hline
    \end{tabular}}
    \caption{Wait-free $m$-valued SW register implementations from $b$-valued SW registers, where $r$ is the number of readers.}
    \label{complexity-table}
\end{table}

Existing implementations of multi-word SW regular or atomic registers from base registers of the same type fall into two broad categories:
ones where all readers are invisible (meaning they do not write to any base registers) and ones where all readers are visible.
The invisible-reader implementations in Table~\ref{complexity-table} all use at least $\Omega(\frac{m}{b})$ space, which is exponentially higher than the trivial information-theoretic lower bound of $\Omega(\frac{\log m}{\log b})$.
Researchers have improved the constant factors of this lower bound~\cite{chaudhuri1994bounds}, but asymptotic improvements have only been shown for restricted classes of algorithms~\cite{BergerIntegrated,chaudhuri2000one}.
A similar but smaller gap exists in the visible reader case.
Since all readers write to a different SW base register, 
it is easy to show a space lower bound of $\lceil r + \frac{\log m}{\log b}\rceil$ in this setting. 
However, the visible-reader implementations in Table~\ref{complexity-table} all use at least $\Omega(r \frac{\log m}{\log b})$ space, and a natural question is whether this is inherently required.

In this paper, we close both gaps by exponentially strengthening the lower bound in the invisible reader setting and improving the upper bound in the visible reader setting.
Specifically, we prove that any algorithm with at least one invisible reader requires at least $\lceil \frac{m-1}{b-1} \rceil$ base registers. This lower bound holds even when implementing a regular register from atomic base registers, so it also holds in the regular-from-regular and atomic-from-atomic cases. Our lower bound is proven for the weaker obstruction-free progress property, which means it also holds for lock-free and wait-free implementations. 
This lower bound is asymptotically tight because there exist wait-free regular-from-regular~\cite{chaudhuri1994bounds} and atomic-from-atomic~\cite{VIDYASANKAR1991323} implementations that use $\Theta(\frac{m}{b})$ space.
In the setting with visible readers, we develop a wait-free algorithm for the atomic-from-atomic case that uses only $\Theta(r + \frac{\log m}{\log b})$ base registers. Chen and Wei's visible-reader algorithm uses $\Theta(r\frac{\log m}{\log b})$ space and $\Theta(\frac{\log m}{\log b})$ steps per read and write~\cite{chen2017step}, whereas our algorithm uses $\Theta(r+\frac{\log m}{\log b})$ space, $\Theta(\frac{\log^2 m}{\log^2 b})$ read steps, and $\Theta(r+\frac{\log m}{\log b})$ write steps. Whether the optimal $\Theta(r+\frac{\log m}{\log b})$ space bound can be achieved together with $\Theta(\frac{\log m}{\log b})$ read and write step complexity remains open.

Putting together the case with at least one invisible reader and the case in which every reader is visible, we show a lower bound of \hbox{$\lceil \min(\frac{m-1}{b-1}, r+\frac{\log{m}}{\log{b}}) \rceil$} and a matching upper bound of $\Theta(\min(\frac{m}{b}, r+\frac{\log{m}}{\log{b}}))$ on the space complexity of implementing multi-word SW registers. The lower bound holds in the weaker regular-from-atomic setting and it holds for all obstruction-free implementations. The upper bound is wait-free and holds in the atomic-from-atomic setting.
It can be derived by combining our algorithm with Vidyasankar's~\cite{VIDYASANKAR1991323} plus Chaudhuri and Welch's~\cite{chaudhuri1994bounds} algorithm from Table~\ref{complexity-table}, and using whichever requires less space.
The invisible-reader construction uses $\Theta(\frac{m}{b})$ space, whereas our visible-reader construction uses $\Theta(r+\frac{\log m}{\log b})$ space. Therefore, the former gives the smaller asymptotic space bound when $\frac{m}{b}$ is smaller, and the latter does so when $r+\frac{\log m}{\log b}$ is smaller.

\textbf{Outline.} The remainder of the paper is organized as follows. Related work is described in Section~\ref{sec:related}. The model of computation and important terminology are defined in Section~\ref{sec:model}. Our lower bound is proven in Section~\ref{sec:lb} and the algorithm for our upper bound is described in Section~\ref{sec:upper}. Finally, we conclude in Section~\ref{sec:concl}.

\section{Related Work}
\label{sec:related}

Table~\ref{complexity-table} summarizes the step and space complexity of known implementations of an $m$-valued SW register from $b$-valued base registers of the same type; all of them are wait-free. We organize prior work into the space upper bounds these implementations achieve and the lower bounds known for the problem, relating each to our results.
 
\textbf{Prior upper bounds.}
In the invisible-reader setting, Chaudhuri and Welch~\cite{chaudhuri1994bounds} build a tree over the represented values to implement a multi-word regular register from regular base registers. 
Their construction was first given for $b = 2$, and Chen and Wei~\cite{chen2017step} later generalized it to arbitrary $b \geq 2$. When $m$ is a power of $b$, it uses $\frac{m-1}{b-1}$ base registers, matching our lower bound exactly. For the atomic case, Vidyasankar~\cite{VIDYASANKAR1991323} implements an $m$-valued atomic register from two $m$-valued regular registers and one binary atomic register; instantiating the two regular registers with the algorithm of Chaudhuri and Welch~\cite{chaudhuri1994bounds} yields an $m$-valued atomic register from $b$-valued atomic registers using $2\frac{m-1}{b-1} + 1$ space when $m$ is a power of $b$, which matches our lower bound asymptotically. In the visible-reader setting, the best previously known space bound is $\Theta(r \frac{\log m}{\log b})$, achieved by Peterson~\cite{Peterson} and by Chen and Wei~\cite{chen2017step}; our visible-reader algorithm improves this to $\Theta(r + \frac{\log m}{\log b})$.

\textbf{Prior lower bounds.}
The only previously known lower bound that applies to all implementations is the trivial $\Omega(\frac{\log m}{\log b})$ bound, refined by Chaudhuri and Welch~\cite{chaudhuri1994bounds} to $\lceil \max(\log m + 1, (2\log m) - (\log \log m) - 2) \rceil$ for regular registers built from binary regular registers. This is exponentially smaller than the $\Omega(\frac{m}{b})$ space used by every known invisible-reader implementation, and closing this gap is the main contribution of our lower bound. Stronger bounds were known only for restricted classes of algorithms. In the case of algorithms in which each simulated write performs only a single base-register write, Chaudhuri, Kosa and Welch~\cite{chaudhuri2000one} prove two space lower bounds. The first is a $2m - 1 - \lfloor \log m \rfloor$ lower bound for one-write implementations satisfying the \emph{symmetric property}, meaning that whenever a write changing the logical value from $v$ to $w$ modifies a base register $x$, if the subsequent write changes the value back from $w$ to $v$, then the write restores $x$ to its original contents.
For a more restricted class of one-write implementations satisfying the \emph{toggle} property where, for each unordered pair of logical values, there is a unique register that is changed when switching between the two logical values, they show that $\Omega(m^2)$ space is needed, establishing that their one-write algorithm is space-optimal for this class.
Berger, Keidar and Spiegelman~\cite{BergerIntegrated} consider implementations in which each simulated read must observe at least $\tau \geq 2$ values written by the same simulated write before it may return. They prove that any wait-free regular-from-atomic implementation in this class requires at least $\tau m$ space with an invisible reader and $\tau + (\tau-1)\min(m-1, r)$ space in general. Their bound helps explain the space usage of the visible-reader implementations of Peterson~\cite{Peterson} and Chen and Wei~\cite{chen2017step}, for which $\tau = \frac{\log m}{\log b}$. However, every invisible-reader implementation in Table~\ref{complexity-table} has $\tau \leq 1$, so this bound does not apply to any of them; in contrast, our $\lceil \frac{m-1}{b-1} \rceil$ bound holds for all invisible-reader implementations.

 
\textbf{Other models.}
Some papers~\cite{Larsson} assume stronger atomic primitives such as swap, fetch-and-add, and compare-and-swap; we focus on implementations that use only read/write base registers.

    \section{Model}
\label{sec:model}



We will work in the standard asynchronous shared memory model \cite{attiya2004distributed} with $r$ readers and one writer, which communicate through shared base registers. Unless otherwise stated, all base registers are atomic. Processes may fail by crashing, and each process invokes at most one operation at a time.


An \emph{execution} is an alternating sequence of \emph{configurations} and \emph{steps} $C_0, e_1, C_1, e_2, C_2, \dots$, beginning with an \emph{initial configuration} $C_0$. Each step $e_i$ is a read or write of a single base register performed by one process, and the configuration $C_i$ records the state of every base register together with the local state of every process after $e_i$ is applied to $C_{i-1}$. Given two configurations $C$ and $C'$ of the same execution with $C$ no later than $C'$, the \emph{interval} from $C$ to $C'$ consists of all configurations and steps from $C$ through $C'$, inclusive.



An operation is \emph{complete} if it has returned a response, and \emph{pending} otherwise. The \emph{execution interval} of a complete operation is the set of all configurations and steps from its first step to the configuration immediately after its last step; for a pending operation, the execution interval consists of all configurations and steps from its first step onward. Two operations are \emph{concurrent} if their execution intervals intersect.



A register is \emph{atomic} if its operations are linearizable~\cite{HW}: for every execution, each completed operation can be assigned a \emph{linearization point} --- a configuration within its execution interval --- so that all operations behave as if they occurred sequentially in the order of their linearization points. When we say that an operation is \emph{linearized at a step}, we mean that its linearization point is the configuration immediately after that step.
 
A register is \emph{regular} if every read returns either the initial value if no write completed before the read's first step, the value written by the last write that completed before the read's first step, or the value written by some write concurrent with the read. Every atomic register is also regular.



An implementation is \emph{obstruction-free} if, from any reachable configuration, an operation that runs in isolation --- with all other processes suspended --- terminates in a finite number of steps. It is \emph{lock-free} if, in every infinite execution, infinitely many operations complete, and \emph{wait-free} if every operation completes in a finite number of its own steps, regardless of the steps taken by other processes. Every wait-free implementation is lock-free, and every lock-free implementation is obstruction-free.


The \emph{step complexity} of an operation is the maximum number of base-register accesses it performs, over all executions. The \emph{space complexity} of an implementation is the number of base registers it uses.



    \section{Space Lower Bound}
\label{sec:lb}
This section proves new lower bounds on the number of atomic base registers needed to implement a multi-word regular register. Throughout, we abbreviate ``obstruction-free implementation of a regular SW register from atomic SW base registers'' to simply \emph{implementation}. The first bound, Theorem~\ref{thm:lb_inv}, applies whenever at least one reader is invisible, including implementations that also contain visible readers. The second, Theorem~\ref{thm:lb_vis}, extends it to all implementations by combining this case with the case in which every reader is visible.

We begin by introducing the notions used in the proofs. Throughout this section, let $m$ be the number of values of the simulated register and let $N$ be the number of base registers of the implementation under consideration. Let $b_i$ be the number of values that the $i^\text{th}$ base register can represent.
The \emph{total fanout} of the implementation, denoted $S$, is the sum $\sum_i b_i$ of these quantities. When every base register is $b$-valued we have $S = Nb$, but it is convenient in the proof to allow base registers of different sizes.

The intuition behind the term total fanout comes from viewing the invisible reader's algorithm as a decision tree. 
The internal nodes of the decision tree are labeled by base registers and the same base register can appear multiple times in the decision tree.
The leaves are labeled by return values.
When a reader reads the $i$-th base register, the algorithm can branch in $b_i$ different directions based on the value read; thus, $b_i$ represents the branching factor, or \emph{fanout} of that register. 

\subsection{Proof Overview}
The key step is Lemma~\ref{lem:lb_main}, which transforms an $m$-valued implementation with total fanout $S$ and an invisible reader into an $(m-1)$-valued implementation with total fanout $S-1$. We later apply this lemma inductively to argue that $S$ must be large whenever $m$ is large.
To prove Lemma~\ref{lem:lb_main}, we fix one invisible reader and represent its algorithm as a decision tree.
We select a minimum-depth leaf $\ell$ returning a value $v$. Let $p$ be its parent, let $t$ be the base register read at $p$, and let $x$ be the value labeling the edge from $p$ to $\ell$. If returning $v$ is legal whenever the reader is poised at $p$, we simplify the algorithm by replacing the subtree rooted at $p$ by a leaf returning $v$. Repeating this simplification eventually yields an execution in which a pending read $R$ is poised at $p$, but returning $v$ would violate regularity.

From this execution, we complete every other pending operation while leaving $R$ paused, resulting in a configuration $C$. Since $v$ is not a legal response for $R$ at this point due to regular-register semantics, no write of $v$ is pending or concurrent with $R$. Since no new operations are invoked, no write of $v$ occurs while the other pending operations are completed, and $t$ cannot take the value $x$ during this interval; otherwise, resuming $R$ immediately after $t$ takes on the value $x$ would make it return the illegal value $v$. We then run $R$ to completion to obtain a configuration with no pending operations. Because the reader is invisible, this configuration has the same shared memory and the same local states for every other process as the preceding configuration in which $R$ was still poised at $p$.

We use this resulting configuration, in which all pending operations including $R$ have completed, as the initial configuration of an implementation whose simulated domain excludes $v$. If $t$ could ever contain $x$ in an execution of this $(m-1)$-valued register implementation, we could remove the invisible reader's steps, replay the remaining execution from the preceding configuration $C$ in which all other pending operations have completed and $R$ is still poised at $p$, and then resume $R$. It would read $x$ and return $v$, contradicting regularity. Thus $x$ is never stored in $t$ in the $(m-1)$-valued implementation and can be removed from $t$'s domain, decreasing total fanout by one.
Thus, we obtain an $(m-1)$-valued implementation with $N$ base registers, total fanout $S-1$, and an invisible reader, as claimed.
Lemma~\ref{lem:inv} applies this reduction inductively to show that $S-N+1\geq m$. The invisible-reader bound, Theorem~\ref{thm:lb_inv}, proves that when all base registers are $b$-valued, $S=Nb$ yields $N \geq \lceil \frac{m-1}{b-1} \rceil$. If every reader is visible, the readers require $r$ distinct base registers and the writer requires at least $\lceil \frac{\log{m}}{\log{b}} \rceil$ additional registers. The general lower bound, Theorem~\ref{thm:lb_vis}, follows by combining Theorem~\ref{thm:lb_inv} with the lower bound for implementations in which every reader is visible.

    \subsection{Lower Bound Proof}
    Before diving into the main technical lemma, we define a useful predicate. 

\begin{definition}
\label{def:e_pred}
    The predicate $E(m, N, S)$ holds if there exists an obstruction-free implementation of an $m$-valued regular SW register using $N$ atomic SW base registers with total fanout $S$ in which at least one reader is invisible.
\end{definition}
    
    
    \begin{lemma}
    \label{lem:lb_main}
        For $m \geq 2$, $E(m, N, S)$ implies $E(m-1, N, S-1)$.
    \end{lemma}

\begin{proof}
Without loss of generality, we represent the algorithm of an invisible reader as a set of decision trees, one for each possible local state at the start of a read operation, since the reader may retain local state between operations. Each internal node represents the reader's complete local state immediately before its next base-register read and is labeled by that base register. Each such node has $b_i$ outgoing edges, where $i$ is the index of the base register labeling the node, and each outgoing edge is labeled by a value that can be read from the base register. The invisible reader chooses a decision tree based on its local state and starts at the root of that decision tree. It then repeatedly performs the following steps until it reaches a leaf: read the register labeling the current node and follow the edge corresponding to the value that was read. After following an edge from node $n_1$ to node $n_2$, we say that the reader has \emph{reached} node $n_2$. Each leaf is labeled by a value from the simulated register and by the complete local state of the reader after returning that value. Once the invisible reader reaches a leaf, it returns the value labeling that leaf and enters the corresponding local state. We say that an invisible reader is \emph{poised at} a node in its decision tree if it has reached that node but has not yet performed the corresponding read.

For a pending read $R$ after an execution $E$, we say that $R$ can return a value $v$ immediately after $E$ if returning $v$ at that point would satisfy regular-register semantics; that is, $v$ is the initial value if no write completed before $R$ began, the value of the last write that completed before $R$ began, or the value of a write concurrent with $R$.

We say that a process \emph{runs solo} during an interval of an execution if no other processes take any steps in that interval.

Suppose $E(m, N, S)$ holds, witnessed by an implementation $A_{m}$ satisfying Definition~\ref{def:e_pred}. 
Let $IV$ be one of the invisible readers in $A_m$. Fix the initial configuration of $A_m$, let $s$ be the local state of $IV$ in that configuration, and consider the decision tree $D$ for a read invoked by $IV$ from local state $s$. To simplify $A_m$, we begin by changing the post-return local state of every leaf in $D$ to $s$ and we call the resulting implementation $A_m'$.
This change forces $IV$ to use $D$ for every read operation.

We first show that $A_m'$ still implements a regular register. Suppose otherwise, and let $E$ be an execution prefix of $A_m'$ ending when a read $R$ returns a value that violates regular-register semantics. If $R$ is performed by $IV$, let $E'$ be obtained from $E$ by removing all earlier completed reads by $IV$, but retaining $R$. Since every read by $IV$ in $A_m'$ leaves it in state $s$, the read $R$ begins in state $s$. After removing the earlier reads, $IV$ also begins $R$ in state $s$ in $A_m$ and therefore follows the same decision tree $D$. If $R$ is performed by another reader, let $E'$ be obtained from $E$ by removing all steps by $IV$. In either case, since $IV$ is invisible, the removed steps do not change shared memory or the local states of any other process. Therefore, there is a corresponding execution prefix of $A_m$ in which the same base-register steps occur as in $E'$, and $R$ reads the same base-register values and returns the same value. The writer steps and their order relative to $R$ are unchanged. Therefore, the initial value if no write completed before $R$ began, or otherwise the value of the last such write, and the values written by all writes concurrent with $R$ are unchanged. Thus, $R$ violates regular-register semantics in $A_m$, contradicting the correctness of $A_m$. Therefore, $A_m'$ implements a regular register.

We next show that $A_m'$ is also obstruction-free. Consider any configuration $C$ reachable in $A_m'$, and fix an execution of $A_m'$ ending in $C$. If $IV$ has a pending read, remove from this execution all of its earlier completed reads while retaining the pending read. Otherwise, remove all of its completed reads. Since every completed read by $IV$ in $A_m'$ leaves it in state $s$ and $IV$ is invisible, the resulting sequence is an execution of $A_m$. If $IV$ has a pending read, that read began in state $s$, and the two implementations behave identically until the read returns. Thus, the resulting sequence reaches the same complete configuration $C$ in $A_m$. Therefore, any operation running solo from $C$ takes the same steps until it completes as it would in $A_m$, and hence terminates after finitely many steps. For the remainder of the proof, we write $A_m$ for $A_m'$.


Consider a leaf $\ell$ of minimum depth in the decision tree $D$ and let $v$ be its value. The decision tree is fixed by $IV$'s algorithm and the chosen starting local state $s$: the actions of other processes do not change the tree, but may change the base-register values observed by $IV$ and therefore the path it follows. Although the tree may have infinite depth and may contain infinite-length paths, it has at least one leaf at finite depth: let $IV$ invoke the read from the initial configuration, suspend all other processes, and let $IV$ run alone. The values returned by its base-register reads determine one path through the tree. By obstruction-freedom, this solo read terminates after finitely many of its own steps, so this path reaches a leaf at finite depth. Thus, at least one leaf has finite depth, and consequently the minimum leaf depth is finite. Moreover, $\ell$ is not the root: starting from the initial configuration, consider two executions in which $IV$ takes no steps: in one, the writer runs solo and completes a write of some value $a$, and in the other, the writer runs solo and completes a write of a different value $b$. In both executions, $IV$ subsequently invokes a read from the same local state $s$. If the root were a leaf, the read would return the same value in both executions, contradicting regular-register semantics in at least one of them. Let $p$ be the parent of $\ell$. If $IV$ can return $v$ immediately after every execution in which it is poised at $p$, replace the subtree rooted at $p$ with
a leaf labeled by $v$ and the post-return local state $s$. This preserves correctness because every newly introduced return of $v$ is permitted by regular-register semantics, $IV$ does not modify shared memory, and its post-return local state is
the initial local state of IV.
It preserves obstruction-freedom because it only shortens reads. The minimum leaf depth decreases by one. Repeating this process must terminate because the minimum leaf depth is a nonnegative integer, and it cannot decrease to zero by the preceding argument. Continue to denote the resulting implementation by $A_m$. We therefore obtain a leaf $\ell$ with value $v$, its parent $p$, and an execution $E_0$ ending in a configuration $C_0$ after which a pending read $R$ by $IV$ is poised at $p$ and cannot return $v$ immediately without violating regular-register semantics.

Since $R$ cannot return $v$ immediately after $E_0$, if no write completed before $R$ began, then the initial value is not $v$; otherwise, the last write that completed before $R$ began did not write $v$. Moreover, no write of $v$ is concurrent with $R$ in $E_0$. In particular, there is no pending write of $v$ at $C_0$. Extend $E_0$, without invoking any new operations, by running each process with a pending operation, except for the pending read $R$, solo until its operation completes, which it does after a finite number of steps by obstruction-freedom. Let $E_1$ denote the resulting execution and let $C_1$ denote its final configuration. No write of $v$ occurs between $C_0$ and $C_1$, so $R$ still cannot return $v$ at $C_1$.

Let $x$ be the value labeling the edge from $p$ to $\ell$, and let $t$ be the register labeling $p$. We next show that $t$ does not store $x$ at $C_0$ or in any configuration between $C_0$ and $C_1$. Suppose otherwise. At the first such configuration, we could resume $R$; it would read $x$ from $t$, reach $\ell$, and return $v$. Since no write of $v$ occurred, this would violate regular-register semantics. In particular, $t \neq x$ at $C_1$.

Extend $E_1$ by running $R$ solo until it completes, which it does after a finite number of steps by obstruction-freedom. Let $E_2$ denote the resulting execution and let $C_2$ denote its final configuration. Because $IV$ is invisible, $C_1$ and $C_2$ contain the same values in every base register and the same local states for every process other than $IV$. In particular, $t \neq x$ at $C_2$. Every process has no pending operation at $C_2$. Let $w_0$ be the value of the last write that completed before $C_2$, or the original initial value if no write has completed. We have $w_0 \neq v$, since if no write completed before $R$ began, then the initial value was not $v$; otherwise, the last write preceding $R$ did not write $v$. Moreover, no write of $v$ was concurrent with $R$ or occurred while constructing $C_2$.

We define an implementation $A_{m-1}$ of an $(m-1)$-valued regular register by (1) using the same set of base registers as $A_m$, (2) using $C_2$ as the initial configuration and $w_0$ as the initial simulated value, (3) using the same reader algorithms as in $A_m$, and (4) using the same writer algorithm as $A_m$ for all values $w \neq v$. This is a correct implementation of a regular $(m-1)$-valued register because any execution of $A_{m-1}$ can be concatenated after $E_2$ to obtain an execution of $A_m$. Since every process has no pending operation at $C_2$ and its last completed write has value $w_0$, an incorrect execution of $A_{m-1}$ would therefore give an incorrect execution of $A_m$. Obstruction-freedom is also inherited because every configuration reachable in $A_{m-1}$ is reachable in $A_m$, and the processes use the same algorithms.

We argue that $t$ never stores value $x$ in any configuration of any execution of $A_{m-1}$. Suppose for contradiction that some execution $E_3$ of $A_{m-1}$ ends in a configuration $C_3$ in which $t$ stores $x$. Remove from $E_3$ all steps taken by $IV$, resulting in an execution $E_4$. Since $IV$ is invisible, removing its steps does not change the shared-memory states observed by any other process. Furthermore, $C_1$ and $C_2$ have identical shared memory and identical local states for every process other than $IV$. Therefore, the same sequence of steps in $E_4$ is executable from $C_1$ and reaches the same shared-memory configuration; in particular, it reaches a configuration in which $t$ stores $x$.

Concatenating $E_1$ with the same sequence of steps as in $E_4$ produces an execution $E_5$ of $A_m$ in which the original read $R$ remains poised at $p$. If we extend $E_5$ by resuming $R$, it reads $x$ from $t$, reaches $\ell$, and returns $v$. If no write completed before $R$ began, then the initial value is not $v$; otherwise, the last write that completed before $R$ began did not write $v$. Moreover, every write concurrent with $R$ writes a value different from $v$: this was true at $C_0$, no write of $v$ occurred while constructing $C_1$, and $A_{m-1}$ does not permit writes of $v$. Thus, $R$ returning $v$ contradicts the correctness of $A_m$. Therefore, $t$ never stores $x$ in any reachable configuration of $A_{m-1}$.

Since $t$ does not initially store $x$ and never stores $x$ in any reachable execution, we can remove $x$ from the set of values representable by $t$ without changing any reachable execution. This reduces the total fanout by exactly one while preserving correctness and obstruction-freedom. Since $A_{m-1}$ has at least one invisible reader, inherited from $A_m$, it follows that $E(m-1,N,S-1)$ holds.
\end{proof}

The next lemma applies Lemma~\ref{lem:lb_main} inductively to show that $S$ must be large when $m$ is. For intuition, a finite rooted tree with $q$ internal nodes whose branching factors sum to $T$ has $T-q+1$ leaves.

    
    
    \begin{lemma}
    \label{lem:inv}
        $E(m, N, S)$ implies $S - N + 1\geq m$.
    \end{lemma}

    \begin{proof}
        We proceed by induction on $m$. 
        For the base case where $m = 1$, every base register has fanout at least one, so $S \geq N$ and hence $S - N + 1 \geq 1 = m$. For the inductive step, assume the lemma holds for some $m = k$, where $k \geq 1$.
        By Lemma~\ref{lem:lb_main}, $E(k+1, N, S)$ implies $E(k, N, S-1)$ since $k+1 \geq 2$. By the inductive hypothesis $E(k, N, S-1)$ implies $(S - 1) - N + 1 \geq k$. Adding one to both sides gives $S - N + 1 \geq k + 1$, as required.
        
    \end{proof}

\begin{theorem}
\label{thm:lb_inv}
    Any obstruction-free implementation of an $m$-valued regular SW register from $b$-valued atomic base registers in which \textbf{some reader is invisible} uses at least $\lceil \frac{m-1}{b-1} \rceil$ base registers.
\end{theorem}
    
    
    \begin{proof}
Any such implementation using $N$ base registers witnesses $E(m, N, Nb)$.
By Lemma~\ref{lem:inv}, $Nb - N \geq m - 1$, that is, $N \geq \frac{m-1}{b-1}$. Since $N$ is an integer, $N \geq \lceil \frac{m-1}{b-1} \rceil$.
    \end{proof}

\begin{theorem}
\label{thm:lb_vis}
    Any obstruction-free implementation of an $m$-valued regular SW register from $b$-valued atomic base registers uses at least $\lceil \min(\frac{m-1}{b-1}, \, r+\frac{\log{m}}{\log{b}}) \rceil$ base registers, where $r$ is the number of readers.
\end{theorem}

    
    \begin{proof}
    Let $A$ be any obstruction-free implementation of an $m$-valued regular SW register from $b$-valued atomic base registers. If $A$ has an invisible reader, then by Theorem~\ref{thm:lb_inv} it uses at least $\lceil \frac{m-1}{b-1} \rceil$ base registers. Otherwise every reader is visible, and since base registers are single-writer, each reader writes to a different base register, accounting for $r$ of them; the writer needs a further $\lceil \frac{\log{m}}{\log{b}} \rceil$ base registers to represent the $m$ possible values, for a total space usage of at least $\lceil r+\frac{\log{m}}{\log{b}} \rceil$. Putting these two cases together yields the desired lower bound.

    \end{proof}

\section{Space Upper Bound}
\label{sec:upper}

In Section~\ref{sec:lb}, we proved that every obstruction-free
implementation of an $m$-valued regular SW register from $b$-valued atomic
base registers requires at least
$\left\lceil\min\!\left(\frac{m-1}{b-1},\, r+\frac{\log m}{\log b}\right)\right\rceil$
base registers. We now show that this bound is asymptotically tight, even for
the stronger goal of implementing an \emph{atomic} register with a
\emph{wait-free} algorithm. The first term of the minimum is already matched by
known invisible-reader constructions: instantiating Vidyasankar's
regular-to-atomic transformation~\cite{VIDYASANKAR1991323} with the tree-based
regular register of Chaudhuri and Welch~\cite{chaudhuri1994bounds} yields a
wait-free atomic-from-atomic implementation using $\Theta(\frac{m}{b})$ base
registers. This section supplies an upper bound matching the second term. We give a wait-free
implementation of an $m$-valued atomic SW register from $b$-valued atomic SW
base registers in Figure~\ref{alg:atomic} that uses only
$\Theta\!\left(r+\frac{\log m}{\log b}\right)$ base registers. Running whichever of the two implementations is cheaper for
the given parameters then matches the lower bound with a space complexity of
$\Theta\!\left(\min\!\left(\frac{m}{b},\, r+\frac{\log m}{\log b}\right)\right)$,
improving on the previous best space bound of
$\Theta\!\left(\min\!\left(\frac{m}{b},\, r\frac{\log m}{\log b}\right)\right)$, obtained by combining the invisible-reader construction above with the visible-reader algorithms of Peterson and Chen and Wei~\cite{Peterson,chen2017step}.
The space
savings come at a modest cost in step complexity: in our implementation each
write completes in $\Theta\!\left(r+\frac{\log m}{\log b}\right)$ steps and each
read in $\Theta\!\left(\frac{\log^2 m}{\log^2 b}\right)$ steps.

At a high level, the algorithm combines the double buffering technique with helping: a reader first tries to scan the active global buffer directly, and if it fails this scan too many times, it returns the value that the writer delivered one piece at a time.
Between every two failed scan attempts on the global buffer, the reader is guaranteed to receive a piece of a value that is safe to return from the writer.

We first define notation and
conventions for the implementation in Section~\ref{sec:notation} and then give a detailed description of the algorithm in Section~\ref{sec:overview}. The remaining subsections prove that the algorithm is linearizable and wait-free with the desired bounds.

\begin{figure}
{\small
\begin{algorithmic}[1]
\Statex \textbf{Shared Variables:}
\Statex $L = \lceil \frac{\log m}{\log b} \rceil$: number of pieces per value
\Statex $G[0], G[1]$: arrays of $L$ registers
\Statex $V$: 1-bit pointer register (active buffer pointer)
\Statex For each reader $i \in \{1, \dots, r\}$:
\Statex \quad $Mailbox_i$: 1 register (holds 1 piece of the value)
\Statex \quad $ReadReq_i, ReadAck_i$: 1-bit handshake for new read operation
\Statex \quad $AttemptReq_i, AttemptAck_i$: 1-bit handshake for interruption tripwire
\Statex \quad $PieceAck_i, PieceReady_i$: 1-bit handshake for piece transfer

\Statex \textbf{Writer Local Variables:}
\Statex $localVal[1 \dots r]$: one $L$-piece value per reader
\Statex $idx[1 \dots r]$: integer $0 \dots L$ (next piece to send)
\Statex
\Statex \textbf{Initialization:}
\Statex $V \gets 0$, and both $G[0]$ and $G[1]$ contain the initial value
\Statex For each reader $i$: Set every handshake bit to 0, and $idx[i] \gets L$
\Statex
\makeatletter
\setcounter{ALG@line}{-1}
\Procedure{Write}{$value$}
    \State $s \gets 1 - \text{read}(V)$ \Comment{Inactive buffer}
    \State $\text{write}(G[s], value)$ \Comment{Write data}
    \State $\text{write}(V, s)$ \Comment{Commit active buffer}
    
    \For{$k = 1$ \textbf{to} $r$} \Comment{Check all readers}
        \State $rReq \gets \text{read}(ReadReq_k)$
        \If{$rReq \neq ReadAck_k$} \Comment{Check if a new read has started}
            \State $localVal[k] \gets value$ \Comment{Fix this committed value for reader $k$}
            \State $idx[k] \gets 0$ \Comment{Reset next-piece index}
            \State $\text{write}(PieceReady_k, \text{read}(PieceAck_k))$ \Comment{Clear stale mailbox signal}
            \State $\text{write}(ReadAck_k, rReq)$ \Comment{Acknowledge new read}

        \EndIf
        
        \If{$\text{read}(PieceAck_k) = PieceReady_k$ \textbf{and} $idx[k] < L$} \Comment{Is reader ready?}
            \State $\text{write}(Mailbox_k, localVal[k][idx[k]])$ \Comment{Send piece}
            \State $\text{write}(PieceReady_k, 1 - PieceReady_k)$ \Comment{Signal piece ready}
            \State $idx[k] \gets idx[k] + 1$
        \EndIf
        
        \State $aReq \gets \text{read}(AttemptReq_k)$
        \State $\text{write}(AttemptAck_k, aReq)$ \Comment{Acknowledge tripwire}
    \EndFor
\EndProcedure
\Statex
\setcounter{ALG@line}{-1}

\Procedure{Read$_i$}{}
    \State $myReq \gets 1 - \text{read}(ReadAck_i)$
    \State $\text{write}(ReadReq_i, myReq)$ \Comment{Announce new read}
    \State $pieces \gets []$ \Comment{Empty local list}
    \While{$\text{length}(pieces) < L$}
        \State $aReq \gets 1 - \text{read}(AttemptAck_i)$ \Comment{Arm tripwire}
        \State $\text{write}(AttemptReq_i, aReq)$
        
        \State $s \gets \text{read}(V)$ \Comment{Active buffer}
        \State $val \gets \text{read}(G[s])$ \Comment{Attempt global read}
        
        \If{$\text{read}(AttemptAck_i) \neq aReq$} \Comment{Did writer interrupt?}
            \State \Return $val$ \Comment{No, safe to return}
        \EndIf
        
        \If{$\text{read}(ReadAck_i) = myReq$} \Comment{Wait for writer to sync}
            \If{$\text{read}(PieceReady_i) \neq PieceAck_i$} \Comment{Is piece ready?}
                \State $pieces\text{.append}(\text{read}(Mailbox_i))$ \Comment{Grab piece}
                \State $\text{write}(PieceAck_i, \text{read}(PieceReady_i))$ \Comment{Acknowledge piece}
            \EndIf
        \EndIf
    \EndWhile
    \State \Return $pieces$ \Comment{All pieces received}
\EndProcedure
\Statex
\makeatother
\end{algorithmic}}
\caption{
Wait-Free atomic $m$-valued SW register from atomic $b$-valued SW registers.}
\label{alg:atomic}
\end{figure}

\subsection{Definitions and Initialization}
\label{sec:notation}

For the implementation in Figure~\ref{alg:atomic}, a register value consists of $L$ pieces, where
$
L=\left\lceil \frac{\log m}{\log b} \right\rceil.
$
The global buffers $G[0]$ and $G[1]$ are arrays of $L$ atomic registers. A scan of $G[s]$ means reading from all $L$ registers of $G[s]$ in some fixed order. A write of $G[s]$ means writing to all $L$ registers of $G[s]$ in some fixed order. The order of the pieces is irrelevant for the proof.
A reader returning an $L$-piece sequence returns the value encoded by that sequence.
The pointer register $V \in \{0,1\}$ identifies the active buffer. 

For each reader $i$, there are registers \ReadReq{i}, \ReadAck{i}, \AttemptReq{i}, \AttemptAck{i}, \PieceAck{i}, \PieceReady{i}, and \Mailbox{i}. The writer alone writes to \ReadAck{i}, \AttemptAck{i}, \PieceReady{i}, and \Mailbox{i}. Reader $i$ alone writes to \ReadReq{i}, \AttemptReq{i}, and \PieceAck{i}. 

The writer also maintains local variables \localVal{i} and \idx{i}. For a read request, \localVal{i} stores the fixed $L$-piece value being delivered to reader $i$, and \idx{i} is the index of the next piece to send. We call this ordered delivery a \emph{stream}. A value written to \Mailbox{i} is a \emph{mailbox piece}; the condition $\PieceReady{i}\neq\PieceAck{i}$ is a \emph{mailbox signal} indicating that an unacknowledged piece is ready, and such a signal is \emph{stale} if it belongs to an earlier stream.

The initial configuration satisfies $\ReadReq{i}=\ReadAck{i}$, $\AttemptReq{i}=\AttemptAck{i}$, $\PieceAck{i}=\PieceReady{i}$, and $\idx{i}=L$ for every reader $i$, and both global buffers contain the initial value of the simulated register. The condition $\idx{i}=L$ ensures that the writer sends no mailbox piece to reader $i$ before it has captured a value for some read request.

For a multi-word write operation $W$, let $c(W)$ denote the configuration immediately after $W$'s write to $V$. The write to $V$ is called the \emph{commit step} of $W$. After $c(W)$, the writer scans the readers once. For reader $i$, the part of this scan that inspects and updates the registers associated with $i$ is called the \emph{service block} of $W$ for reader $i$. 
A \Write{} operation does not return until after it has completed all of its service blocks. Since the writer process invokes at most one operation at a time, no later \Write{} operation begins before the previous \Write{} operation has completed all of its service blocks. During the service blocks of a write $W$, the \emph{current committed value} is the value of $W$, since $W$ has already executed its commit step and no later write has begun. When $W$ acknowledges a new read request from reader $i$, it fixes this value for the read by storing it in \localVal{i}.



\subsection{Algorithm Overview}
\label{sec:overview}

The algorithm combines two standard ideas: double buffering and per-reader
assistance. Double buffering lets the writer prepare a new value without
overwriting the current value.
Per-reader assistance allows reads to eventually return a correct value after enough failed attempts at reading the active buffer.
At any
configuration, the active buffer is indicated by the bit $V$. To write a new value, the writer fills the inactive buffer
completely and only then changes $V$ to point to that buffer; the configuration immediately after this update to $V$ is the linearization point of the write. This ensures that the active
buffer $G[V]$ always contains the value of the most recently linearized write.


After changing $V$, the writer executes one service block for each reader. These service blocks
allow slow readers to make progress even if their direct buffer scans keep
failing. For each reader $i$, the writer first checks whether reader $i$ has
announced a new read by comparing its request bit $rReq :=\ReadReq{i}$ with \ReadAck{i}. If they differ,
the writer treats this as a new request: it fixes the current committed value
as the value that will be streamed to reader $i$, stores that value locally as
\localVal{i}, resets \idx{i} to the first piece, and clears any stale mailbox signal by making \PieceReady{i} agree with \PieceAck{i}. This clearing step
ensures that the reader will not confuse a piece from an old mailbox stream
with a piece from the new one. The writer then acknowledges the request by
copying the request bit into \ReadAck{i}. After handling the read request, the
writer may send one mailbox piece: if \PieceAck{i}=\PieceReady{i}, meaning there is no unacknowledged mailbox
piece waiting for the reader, and there are still pieces left to send, the
writer writes the next piece into \Mailbox{i}, flips \PieceReady{i} to announce that a new piece is
ready, and advances \idx{i}. Finally, the writer triggers the reader's
tripwire by copying \AttemptReq{i} into \AttemptAck{i}.

A read by reader $i$ begins by creating a fresh read request. The reader reads
its current acknowledgement bit \ReadAck{i}, chooses the opposite bit as its
new request bit $myReq$, writes $myReq$ to \ReadReq{i}, and starts with an
empty list of mailbox pieces. The reader then first tries to read directly
from the active global buffer. This is the optimistic, fast-path case. The
difficulty is that a reader scans the buffer piece by piece, so the writer
might switch buffers and later return to overwrite the same buffer while the
reader is still scanning it. To detect this kind of interference, reader $i$
uses a small tripwire formed by \AttemptReq{i} and \AttemptAck{i}. Before scanning, the reader reads \AttemptAck{i}, chooses the opposite bit as
its fresh attempt bit $aReq$, and writes $aReq$ to \AttemptReq{i}. If the
writer makes enough progress to possibly endanger the scan, then during its 
service block of reader $i$ it copies this attempt bit into \AttemptAck{i}.
Thus, after the scan, if the reader sees that the tripwire was not triggered,
it knows the scan was safe and can return the value it read from the global
buffer.

If the tripwire is triggered, the reader does not trust the global scan.
Instead, it relies on assistance from the writer. Each reader has a small
mailbox, \Mailbox{i}, through which the writer can deliver one piece of a
value at a time. The reader first checks whether its read request has been
acknowledged, by testing whether \ReadAck{i} equals its request bit $myReq$.
If the request has not yet been acknowledged, then the writer has not yet
fixed a mailbox stream for this read, so the reader simply tries another
iteration. If the request has been acknowledged, then the writer has fixed the
current committed value for this read and stored it locally as \localVal{i}.
From that point on, the writer streams the pieces of that fixed value through
the reader's mailbox. The writer and reader use a one-bit handshake,
\PieceReady{i} and \PieceAck{i}, to ensure that a new mailbox piece is not
sent until the previous one has been received. When the two bits differ, the
reader knows that a new piece is ready: it reads the piece from \Mailbox{i},
appends it to its local list, and then acknowledges the piece by copying
\PieceReady{i} into \PieceAck{i}. This acknowledgement allows the writer to
send the next piece later. The handshake prevents pieces from being skipped,
duplicated, or overwritten before the reader consumes them.

Thus a read has two possible outcomes. If it experiences little or no
interference, it succeeds quickly by returning a direct scan of the active
buffer $G[V]$. If it repeatedly overlaps writes, then the writer's service
blocks gradually deliver the fixed value \localVal{i} through the mailbox.
In either case, the reader eventually obtains a complete value: either
directly from a stable global buffer, or piece by piece from the writer's
fixed mailbox stream.


\subsection{Linearization Points}
\label{sec:linpoint}

\begin{itemize}
    \item \textbf{Writes:} 
    A write operation $W$ that executes its commit step is linearized at configuration $c(W)$, the configuration immediately after its write to $V$ (i.e., at line 3 of $W$).

    \item \textbf{Fast-Path Reads:} 
    A read operation $R$ that returns through the fast path is linearized at the configuration immediately after the read of $V$ in the successful fast-path iteration (i.e., at line 7 of $R$ if $R$ returns at line 10).

    \item \textbf{Slow-Path Reads:} For a read operation $R$ by reader $i$, define $myReq := 1-\ReadAck{i}$, where \ReadAck{i} is the value read by $R$ at its first step. The read then writes $\ReadReq{i}:=myReq$. If $R$ returns through the slow path, then $R$ is linearized at the configuration immediately after the first writer step during $R$ that writes $\ReadAck{i}:=myReq$ (i.e., at the first writer's line 10 during $R$ that writes $\ReadAck{i}:=myReq$). This step exists for every slow-path read, as shown below.
\end{itemize}


Pending operations are handled as follows. A pending write operation that has already executed its commit step is included and linearized at the corresponding configuration $c(W)$. A pending write that has not executed its commit step may be discarded. Pending reads may be discarded.


\subsection{Proof Overview}
\label{sec:intuition}

The main correctness idea is that a reader either obtains a clean direct scan
of the active global buffer, or else obtains all pieces of one fixed
writer value through its mailbox.

To prove linearizability, it suffices to prove that every completed read returns the value of the latest write linearized before the read's assigned linearization point.
The first key invariant is that, at every configuration, the active buffer
$G[V]$ stores the value of the latest write that has linearized. This
invariant holds because the writer first fills the inactive buffer with the
new value and only then changes $V$ to point to that buffer. Therefore, before
the update to $V$, the old active buffer still stores the latest linearized
write value; after the update to $V$, the new active buffer stores the newly
linearized write value.

For fast-path reads, namely reads whose tripwire is not triggered, the only
danger is that the reader might scan a buffer while the writer is overwriting
that same buffer. The tripwire prevents such a
read from returning. Before scanning, the reader writes a fresh value to
\AttemptReq{i}. If the writer performs enough work to switch buffers and come
back to overwrite the scanned buffer, then the writer must execute a service block for reader $i$ in between. During that service block, it copies
\AttemptReq{i} into \AttemptAck{i}. Hence, if the reader finishes its scan and
still sees that \AttemptAck{i} has not been updated to the fresh tripwire
value, then the writer could not have overwritten the scanned buffer during
the scan. The reader can therefore safely return the scanned value.

For slow-path reads, namely reads whose tripwire is triggered until they
collect all $L$ mailbox pieces, the writer explicitly captures a stable value
for the reader. When the writer acknowledges reader $i$'s read request by writing to
\ReadAck{i}, it has already stored the current committed value in \localVal{i}
and reset \idx{i} to the beginning of that value. After this acknowledgement,
the reader's request and acknowledgement bits remain equal until the read
returns, so the writer does not reset \localVal{i} or \idx{i} for that reader
during the read. The mailbox handshake then ensures that the reader receives
the pieces of this fixed \localVal{i} in order. The writer does not send a
new piece until the previous one has been acknowledged, so no piece is
duplicated, skipped, or overwritten before the reader consumes it. Therefore,
if the reader returns through the slow path, it returns exactly the fixed
value captured by the writer at the acknowledgement configuration.

The wait-freedom argument is based on bounding how many failed fast-path
attempts a reader can suffer. Each writer operation has a bounded amount of
work: it writes one global buffer and then scans the fixed set of readers
once. For reads, each failed loop iteration either appends one mailbox piece
or appends nothing. A productive failed iteration appends one piece, so there
can be at most $L$ such iterations. The key point is that two consecutive
failed iterations cannot both be unproductive. If a writer service block causes the second failure, then immediately before the reader checks its mailbox in the second failed iteration,
either a piece was already ready or the writer has just made one ready,
unless all $L$ pieces have already been received. Thus the second failed
iteration must append a piece. Therefore the number of unproductive failed
iterations is at most one more than the number of productive failed
iterations. Since there are at most $L$ productive failed iterations, every read completes
after a bounded number of iterations. More concretely, each read executes at
most $2L+1$ iterations, each costing $O(L)$ reader steps, so every read
completes in $\Theta(L^2)$ reader steps. Similarly, writes complete in $\Theta(L+r)$ writer steps:
the writer fills one $L$-piece buffer and then services $r$ readers,
each service costing $\Theta(1)$ steps.
The $\Theta(r + L)$ bound on space usage can be shown by examination of the pseudo-code.



\subsection{Fast-Path Correctness}
\label{sec:buffer-safety}

\begin{invariant}[Active-buffer invariant]
At every configuration, the active buffer $G[V]$ contains the value of the latest write whose commit step has occurred (i.e., the value of the most recently linearized write
operation), or the initial value if no commit step has occurred.
\end{invariant}

\begin{proof}
Initially, $G[V]$ contains the initial value. Consider a write operation $W$. The writer reads $V$, chooses the other buffer $s:=1-V$, writes to all pieces of $G[s]$, and only then writes $V:=s$. Before $c(W)$, the active buffer is the old buffer $G[V]$ and is not modified by $W$. The step immediately preceding $c(W)$ changes $V$ to $s$, and at configuration $c(W)$, $G[s]$ has already been completely written to with the value of $W$. Therefore, at configuration $c(W)$, the active buffer $G[s]$ contains the value of $W$. While $G[s]$ remains active, the writer can write only to the inactive buffer $G[1-s]$. Thus the active buffer cannot be modified until a later write first commits the other buffer. The invariant follows by induction over commit steps.
\end{proof}

\begin{lemma}[Tripwire lemma]
Consider one loop iteration of \ReadOp{i}. Suppose reader $i$ writes $\AttemptReq{i}:=aReq$, then reads $V=s$, scans $G[s]$, and finally observes $\AttemptAck{i}\neq aReq$. Then the writer writes to no piece of $G[s]$ between the reader's read of $V$ and the end of the scan of $G[s]$.
\end{lemma}

\begin{proof}
Let $C_V$ be the configuration immediately after the reader reads $V=s$, and let $C_E$ be the configuration immediately after the reader completes its scan of $G[s]$. Suppose, for contradiction, that the writer writes to some piece of $G[s]$ between $C_V$ and $C_E$. At configuration $C_V$, buffer $G[s]$ is active. The writer never writes to the active buffer. Therefore, before the writer can write to any piece of $G[s]$, it must first make $G[1-s]$ active by writing $V:=1-s$. Let $W_1$ be the write operation that first writes $V:=1-s$ after $C_V$.

The write operation $W_1$ cannot itself write to $G[s]$ after publishing $G[1-s]$, because its buffer-writing phase precedes its commit step. Hence the first possible write to $G[s]$ after $C_V$ occurs in a later write operation $W_2$. Between the commit step of $W_1$ and the beginning of $W_2$, the single writer completes the reader scan of $W_1$. In the service block for reader $i$, the writer reads \AttemptReq{i} and then writes the value read into \AttemptAck{i}. Reader $i$ wrote $\AttemptReq{i}:=aReq$ before $C_V$ and does not change \AttemptReq{i} again until after the final tripwire check of the same iteration. Since the service block of $W_1$ occurs after $C_V$ and before any write to $G[s]$ by $W_2$, the writer reads $\AttemptReq{i}=aReq$ and writes $\AttemptAck{i}:=aReq$ before the alleged write to $G[s]$ between $C_V$ and $C_E$. Moreover, any later service block before the reader's final tripwire check also reads $\AttemptReq{i}=aReq$ and writes $\AttemptAck{i}:=aReq$, because reader $i$ has not changed \AttemptReq{i} during this iteration.

Therefore, immediately before the reader performs its final tripwire check, \AttemptAck{i} equals $aReq$. This contradicts the assumption that the reader observes $\AttemptAck{i}\neq aReq$. Hence no piece of $G[s]$ is written between $C_V$ and $C_E$.
\end{proof}

\begin{lemma}[Fast-path correctness]
If a read returns through the fast path, then it returns the value of the latest write linearized before its fast-path linearization point.
\end{lemma}

\begin{proof}
Let $C_V$ be the configuration immediately after the successful fast-path iteration reads $V=s$. By the active-buffer invariant, $G[s]$ contains the value of the most recently linearized write at configuration $C_V$, or the initial value if no write has committed. By the tripwire lemma, no piece of $G[s]$ is modified between the read of $V$ and the end of the scan of $G[s]$. Therefore the scan returns exactly the contents of $G[s]$ at configuration $C_V$. The read is linearized at configuration $C_V$. Hence it returns the value of the latest write linearized before its linearization point.
\end{proof}


\subsection{Slow-Path Correctness}

Let $R$ be a read operation by reader $i$, so $myReq$ is the bit written by $R$ to \ReadReq{i} at its start.

\begin{lemma}[Acknowledgement occurs during the read]
If $R$ observes $\ReadAck{i}=myReq$ after its initial read of \ReadAck{i}, then the writer executed a step $\ReadAck{i}:=myReq$ after that initial read and before the observation.
\end{lemma}

\begin{proof}
At the start of $R$, reader $i$ reads the current value of \ReadAck{i} and sets $myReq$ to the opposite bit. Thus immediately after this read, $\ReadAck{i}\neq myReq$. Only the writer writes \ReadAck{i}. Therefore, if reader $i$ later observes $\ReadAck{i}=myReq$, the writer must have written $\ReadAck{i}:=myReq$ after the read of \ReadAck{i} at the beginning of $R$ and before that later observation.
\end{proof}

For a slow-path read $R$, let $A_R$ be the configuration immediately after the first writer step in the execution interval of $R$ that writes $\ReadAck{i}:=myReq$ (i.e., the first writer's line 10 during $R$ that writes $\ReadAck{i}:=myReq$).

\begin{lemma}[Stream initialization]
Let $R$ be a read by reader $i$ that returns through the slow path. Earlier in the same read-request branch, before $A_R$, the writer has set \localVal{i} to the value of the write operation whose service block contains $A_R$, set $\idx{i}:=0$, and flushed the piece handshake by writing $\PieceReady{i}:=\PieceAck{i}$.
\end{lemma}

\begin{proof}
The writer writes to \ReadAck{i} only inside the read-request branch of the service block for reader $i$. In that branch, before writing to \ReadAck{i}, the writer executes the assignments that set \localVal{i} to the value of the write whose service block is being executed, set $\idx{i}:=0$, and write $\PieceReady{i}:=\PieceAck{i}$. Therefore the claim follows directly from the writer's program order.
\end{proof}

\begin{lemma}[No stream reset during the read]
After configuration $A_R$ and before $R$ returns, the writer does not reset \localVal{i} or \idx{i} for reader $i$.
\end{lemma}

\begin{proof}
After $A_R$, we have $\ReadAck{i}=myReq$. We show that from $A_R$ until $R$ returns, $\ReadReq{i}=myReq$ always holds. 

If configuration $A_R$ is reached after $R$ executes its write $\ReadReq{i}:=myReq$, then $\ReadReq{i}=myReq$ from that write until $R$ returns, because reader $i$ invokes no later new read and performs no
later write to \ReadReq{i} during $R$. 

Now suppose configuration $A_R$ is reached before $R$ executes its write $\ReadReq{i}:=myReq$. Let $B$ be the service block in which configuration $A_R$ is reached, and let $C_t$ be the configuration immediately after $B$ reads $rReq:=\ReadReq{i}$. Since at configuration $A_R$ we have $\ReadAck{i}=myReq$, and the writer writes to \ReadAck{i} the value $rReq$
previously read in the same service block, the writer must have read
$rReq=myReq$ at configuration $C_t$. Moreover, the service block enters the stream-initialization branch only if
$rReq\neq\ReadAck{i}$ at the read-request test. Hence, at configuration $C_t$,
$\ReadAck{i}\neq myReq$. From configuration $C_t$ until $A_R$, the writer performs no
other write to \ReadAck{i}, and no later service block can occur before
$A_R$. Therefore $\ReadAck{i}\neq myReq$ throughout the interval from $C_t$
until $A_R$.

We claim that $\ReadReq{i}=myReq$ throughout the interval from $C_t$ until
$A_R$. At configuration $C_t$, this holds because the writer read
$rReq=\ReadReq{i}=myReq$. The only process that writes to \ReadReq{i} is
reader $i$. Any such write is the first write of some read operation by
reader $i$, and it writes the opposite of the value of \ReadAck{i} read at
the beginning of that operation. Since throughout the interval from $C_t$ until
$A_R$ we have $\ReadAck{i}\neq myReq$, every such write to \ReadReq{i} writes
the value $myReq$. Thus \ReadReq{i} cannot change away from $myReq$ before
$A_R$.

Therefore, at configuration $A_R$ we have
$\ReadReq{i}=\ReadAck{i}=myReq$. We now show that this equality persists
until $R$ writes $\ReadReq{i}:=myReq$. We first prove that if configuration $A_R$ is reached before $R$ executes its write $\ReadReq{i}:=myReq$, then configuration $A_R$ is reached after $R$'s initial read of \ReadAck{i} but before that write.

Suppose by contradiction that configuration $A_R$ is reached before $R$'s initial read of \ReadAck{i}. Since configuration $A_R$ lies in the execution interval of $R$, it is reached after $R$ was invoked.
Since reader $i$ has no overlapping read operation and $R$ has not yet
performed any write to \ReadReq{i}, \ReadReq{i} does not change between
$A_R$ and $R$'s initial read. By our discussion above, at configuration $A_R$ we have
$\ReadReq{i}=\ReadAck{i}=myReq$. Thus the writer cannot enter the
stream-initialization branch for reader $i$ again before $R$'s initial read,
because that branch requires $\ReadReq{i}\neq\ReadAck{i}$. Hence $R$'s
initial read would observe $\ReadAck{i}=myReq$ and would choose the opposite
bit, contradicting the definition of $myReq$. Thus configuration $A_R$ cannot be reached before $R$'s initial read of \ReadAck{i}.

It follows that, if configuration $A_R$ is reached before $R$ executes its write
$\ReadReq{i}:=myReq$, then configuration $A_R$ is reached after $R$'s initial read of \ReadAck{i} but before that write. During this part of $R$, reader $i$
performs no write to \ReadReq{i}. Since reader $i$ has no overlapping read
operation and no other process writes \ReadReq{i}, \ReadReq{i} remains equal
to $myReq$ until $R$ writes $\ReadReq{i}:=myReq$. That later write leaves
\ReadReq{i} unchanged.

Therefore, after $A_R$ and before $R$ returns, we always have
$\ReadReq{i}=\ReadAck{i}=myReq$. The writer enters the stream-initialization
branch for reader $i$ only when these two bits differ. Hence the writer does
not enter that branch again before $R$ returns, and so it does not reset
\localVal{i} or \idx{i}.
\end{proof}

\begin{lemma}[Alternating-bit delivery]
No mailbox piece is appended by $R$ before $A_R$. After $A_R$ and before $R$
returns, the sequence of mailbox pieces appended by $R$ is always a prefix of
the fixed array \localVal{i}. If $R$ appends $L$ pieces, then the appended
sequence is exactly the whole array \localVal{i}.
\end{lemma}

\begin{proof}
Before reaching configuration $A_R$, in the same read-request branch, the writer flushes the piece handshake; that is,
it reads the current value of $\PieceAck{i}$ and then performs the flush write
$\PieceReady{i}:=\PieceAck{i}$. We first justify that $\PieceAck{i}$ does
not change between this read of $\PieceAck{i}$ and the flush write.

Let $B$ be the service block in which configuration $A_R$ is reached, and let $C_t$ be the configuration immediately after $B$ reads $rReq:=\ReadReq{i}$. We use the following fact established in the proof of the previous lemma:
$\ReadReq{i}=myReq$ and $\ReadAck{i}\neq myReq$ throughout the interval from
$C_t$ until $A_R$.

Now consider any read operation $S$ by reader $i$ that could write to
\PieceAck{i} after configuration $C_t$ and before $A_R$. Since reader $i$ executes no
overlapping read operations, either $S$ wrote its request bit before configuration $C_t$, or $S$ writes its request bit after configuration $C_t$. In the first case, $S$ has
not completed before configuration $C_t$, since it may later write to \PieceAck{i} after
configuration $C_t$. Since reader $i$ executes no overlapping read operations, no later
read operation by reader $i$ can have started, and hence no later read
operation can have written to \ReadReq{i}, before configuration $C_t$. Therefore the
value read by $B$ from \ReadReq{i} at configuration $C_t$ is exactly the request bit of
$S$. Thus the request bit of $S$ is $myReq$. In the second case, $S$ writes its request bit after configuration $C_t$ and before
$A_R$. Since $\ReadReq{i}=myReq$ throughout this interval, this write to
\ReadReq{i} must write the value $myReq$. Thus the request bit of $S$ is
$myReq$.

Thus, in either case, any read operation $S$ by reader $i$ that could write
to \PieceAck{i} after configuration $C_t$ and before $A_R$ has request bit $myReq$.
But such a write to \PieceAck{i} can occur only after $S$ passes its first
guard, which requires observing $\ReadAck{i}=myReq$. This is impossible
before $A_R$, because $\ReadAck{i}\neq myReq$ throughout the interval from
$C_t$ until $A_R$. Therefore no read operation by reader $i$ writes to
\PieceAck{i} between the flush read and the flush write.

Thus, immediately after the flush write, no mailbox piece is ready for reader
$i$, since $\PieceReady{i}=\PieceAck{i}$. Reader $i$ appends a mailbox piece
only if both guards succeed: it must first observe $\ReadAck{i}=myReq$, and
then observe $\PieceReady{i}\neq\PieceAck{i}$.

After $R$ reads \ReadAck{i} at its first step and before $A_R$, the first
guard is false by the definition of $A_R$. Hence $R$ cannot append any mailbox piece before $A_R$, including a stale piece from an earlier stream.

After the flush write but before the writer deposits the first piece of the
current stream, the second guard is false. Here, depositing a piece means that
the writer writes the piece to \Mailbox{i} and then flips \PieceReady{i}.
The flush write made $\PieceReady{i}=\PieceAck{i}$, and the read operation
$R$ cannot change \PieceAck{i} in this interval: from $R$'s initial read of
\ReadAck{i} until $A_R$, $R$ cannot pass the first guard, and after $A_R$ but
before the first deposit, $R$ cannot pass the second guard. Therefore reader $i$ cannot append a piece for the current stream until the writer has actually deposited one.

After $A_R$, by the no-reset lemma, \localVal{i} and \idx{i} are not reset
before $R$ returns. The writer deposits a piece only when
$\PieceAck{i}=\PieceReady{i}$ and $\idx{i}<L$. In that case, it first writes
$\Mailbox{i}:=\localVal{i}[\idx{i}]$, then flips \PieceReady{i}, and then
increments \idx{i}. After the flip and before reader $i$ acknowledges the piece by writing $\PieceAck{i}:=\PieceReady{i}$, we have
$\PieceAck{i}\neq\PieceReady{i}$. This is because the writer flips
\PieceReady{i} while \PieceAck{i} is unchanged, and reader $i$ has not yet
performed the acknowledgement write. Therefore the writer cannot deposit
another piece and cannot overwrite \Mailbox{i} until that acknowledgement
write occurs.

Reader $i$ appends a piece only when $\PieceReady{i}\neq\PieceAck{i}$. When it
does so, it reads \Mailbox{i}, appends that value, and acknowledges that piece.
Therefore every appended piece is exactly the most recently deposited piece
of this stream, and no deposited piece of this stream is overwritten before
it is consumed. Since the writer
increments \idx{i} exactly once per deposited piece, pieces are deposited in
strictly increasing index order. Thus the sequence of pieces appended by $R$
is always a prefix of the fixed array \localVal{i}. If reader $i$ appends
$L$ pieces, then the appended sequence is exactly all of \localVal{i}.
\end{proof}

\begin{lemma}[Slow-path correctness]
If a read returns through the slow path, then it returns the value of the latest write linearized before its slow-path linearization point.
\end{lemma}

\begin{proof}
Let $R$ be a slow-path read by reader $i$, and let $A_R$ be its linearization point. By the acknowledgement lemma, configuration $A_R$ lies in the execution interval of $R$. Since $R$ appends mailbox pieces only after observing $\ReadAck{i}=myReq$, and since $R$ returns through the slow path only after appending $L$ pieces, configuration $A_R$ is reached before $R$ returns. Configuration $A_R$ is reached in the service block of some write operation $W$. The service block of $W$ occurs after the commit step of $W$, so $W$ has already linearized before $A_R$. Because there is a single writer, no later write can commit before $A_R$: at configuration $A_R$, the writer is still executing the service block of $W$. Hence the latest write linearized before $A_R$ is exactly $W$. By the stream initialization lemma, before $A_R$ in the same
read-request branch, the writer has set \localVal{i} to the value of $W$. By the no-reset lemma, \localVal{i} is not reset before $R$ returns. By the alternating-bit delivery lemma, the $L$ pieces returned by $R$ are exactly the $L$ pieces of that fixed \localVal{i}. Therefore $R$ returns exactly the value of $W$, which is the latest write linearized before $A_R$.
\end{proof}


\subsection{Linearizability}
\label{sec:linproof}

\begin{theorem}
The algorithm is linearizable as an implementation of a simulated atomic SW register.
\end{theorem}

\begin{proof}
Assign linearization points as defined in Section~\ref{sec:linpoint}.
All these points lie inside the corresponding operation intervals. For writes this is immediate. For fast-path reads this is immediate. For slow-path reads, it follows from the acknowledgement lemma and from the fact that the read returns only after observing the acknowledgement and appending $L$ pieces. Because each assigned point lies inside the corresponding operation interval, the induced order also respects the real-time order of non-overlapping operations.

By the fast-path correctness lemma, every fast-path read returns the value of the latest write linearized before its fast-path linearization point. By the slow-path correctness lemma, every slow-path read returns the value of the latest write linearized before its slow-path linearization point.

Therefore, in the sequential order induced by these linearization points, every read returns the value of the latest preceding write, or the initial value if no write precedes it. This is exactly the sequential specification of an atomic register. Hence the implementation is linearizable.
\end{proof}


\subsection{Wait-Freedom and Step Complexity}
\label{sec:wfproof}

Each base-register access counts as one step. Thus scanning a global buffer costs $L$ reader steps, and writing to a global buffer costs $L$ writer steps.

\begin{lemma}[Writes are bounded]
Every \Write{} operation completes in $\Theta(L+r)$ writer steps.
\end{lemma}

\begin{proof}
A write reads $V$, writes $L$ pieces into the inactive buffer, writes to $V$,
and then executes one service block for each of the $r$ readers.

Writing to the inactive buffer costs $\Theta(L)$ steps. Each service block performs
only a constant number of base-register accesses. Thus all service blocks together cost $\Theta(r)$ steps,
and the whole write operation costs $\Theta(L+r)$ writer steps.
\end{proof}

Fix, for the rest of the wait-freedom argument, a read operation $R$ by
reader $i$, and let $myReq$ be the bit written by $R$ to \ReadReq{i} at its
start. A loop iteration of $R$ is called \emph{failed} if it does not return
through the fast path. Equivalently, if the iteration chooses attempt bit
$aReq$, then the final tripwire check observes $\AttemptAck{i}=aReq$. A
failed iteration is called \emph{productive} if it appends one mailbox piece,
and \emph{unproductive} otherwise.

\begin{lemma}[Failed attempts are caused by writer acknowledgements]
If an iteration of $R$ chooses attempt bit $aReq$ and fails, then after the reader's initial read of \AttemptAck{i} in that iteration and before the final tripwire check of that iteration, the writer executes a step $\AttemptAck{i}:=aReq$.
\end{lemma}

\begin{proof}
At the beginning of the iteration, the reader reads \AttemptAck{i} and chooses $aReq$ to be the opposite bit. Therefore, immediately after that read, $\AttemptAck{i}\neq aReq$. The iteration fails only if the later tripwire check observes $\AttemptAck{i}=aReq$. Since only the writer writes to \AttemptAck{i}, the writer must have written $\AttemptAck{i}:=aReq$ between those two reader steps.
\end{proof}

For a failed iteration $I$ of $R$, we say that a service block $B$ for reader
$i$ \emph{causes the failure of $I$} if $B$ contains a write
$\AttemptAck{i}:=aReq$, where $aReq$ is the attempt bit chosen by $I$, and
this write occurs after $I$'s initial read of \AttemptAck{i} and before
$I$'s final tripwire check. By the failed-attempt lemma, every failed
iteration has at least one service block that causes its failure.

\begin{lemma}[Bounding unproductive failed iterations]
Let $P$ be the number of productive failed iterations of $R$, and let $U$ be
the number of unproductive failed iterations of $R$. Then $U \leq P+1$.
\end{lemma}

\begin{proof}
We first show that no two consecutive loop iterations of $R$ can both be
failed and unproductive.

Suppose, for contradiction, that $I_1$ and $I_2$ are consecutive loop
iterations of $R$, and that both are failed and unproductive. By the
failed-attempt lemma, there is a service block $B_1$ for reader $i$ that
causes the failure of $I_1$, and there is a service block $B_2$ for reader
$i$ that causes the failure of $I_2$. Since $B_1$ causes the failure of $I_1$, the write to \AttemptAck{i} in
$B_1$ occurs during $I_1$. Since $B_2$ causes the failure of $I_2$, the write
to \AttemptAck{i} in $B_2$ occurs during $I_2$. The writer is sequential, and
the write to \AttemptAck{i} is the last step of a service block for reader
$i$. Therefore $B_2$ begins after the write to \AttemptAck{i} in $B_1$.
Since that write occurs during $I_1$, the service block $B_2$ begins during
$I_1$ or later. In particular, $B_2$ reads \ReadReq{i} after $R$ has already
written $\ReadReq{i}:=myReq$.

Consider the read-request part of $B_2$, where the writer reads
$rReq:=\ReadReq{i}$ and, if $rReq\neq\ReadAck{i}$, initializes the stream
and writes $\ReadAck{i}:=rReq$. Since this part occurs after $R$
has written $\ReadReq{i}:=myReq$, the writer reads $rReq=myReq$. If the
writer observes $\ReadAck{i}\neq myReq$, then $B_2$ initializes the stream
for $R$ and writes $\ReadAck{i}:=myReq$ before executing its piece-transfer
test, namely the test whether $\PieceAck{i}=\PieceReady{i}$ and
$\idx{i}<L$. If instead the writer observes $\ReadAck{i}=myReq$, then by the
acknowledgement lemma, configuration $A_R$ has already been reached.
In either case, before $B_2$ executes its piece-transfer test, $A_R$ exists
and $\ReadAck{i}=myReq$. Moreover, by the no-reset lemma, after $A_R$ and
before $R$ returns, the writer does not enter the stream-initialization
branch again.

Thus, after $I_2$ fails its final tripwire check, namely the check whether
$\AttemptAck{i}\neq aReq$, it will pass the first slow-path guard
$\ReadAck{i}=myReq$ and reach the piece-ready test, namely the test whether
$\PieceReady{i}\neq\PieceAck{i}$.

We next show that \PieceAck{i} does not change from the beginning of $B_2$
until the piece-ready test of $I_2$. The only process that writes to
\PieceAck{i} is reader $i$, and reader $i$ writes to \PieceAck{i} only after
appending a mailbox piece. The service block $B_2$ begins during $I_1$ or
later, and its write to \AttemptAck{i} occurs during $I_2$, before the final
tripwire check of $I_2$. Since $I_1$ and $I_2$ are consecutive loop
iterations and both are unproductive, reader $i$ appends no mailbox piece
from the beginning of $B_2$ through the piece-ready test of $I_2$. Therefore
\PieceAck{i} does not change during this interval. Let $PA$ denote this
constant value of \PieceAck{i}.

Now consider the piece-transfer test of $B_2$.

First suppose that, at this test, $\PieceReady{i}\neq PA$. Then a mailbox
piece is already ready for reader $i$, and $B_2$ does not deposit another
piece, because the deposit condition requires
$\PieceAck{i}=\PieceReady{i}$. Since \PieceAck{i} remains equal to $PA$, and
since the writer does not enter the stream-initialization branch again before
$R$ returns, no later service block before the piece-ready test of $I_2$ can
make this ready piece disappear or overwrite it. So, any later
piece-transfer test also sees $\PieceAck{i}\neq\PieceReady{i}$ and therefore
cannot deposit another piece.

Next suppose that, at the piece-transfer test of $B_2$,
$\PieceReady{i}=PA$ and $\idx{i}<L$. Then $B_2$ deposits the next piece: it
writes the piece to \Mailbox{i}, flips \PieceReady{i}, and increments
\idx{i}. Since \PieceAck{i} remains equal to $PA$, after the flip we have
$\PieceReady{i}\neq\PieceAck{i}$. As in the previous case, no later service
block before the piece-ready test of $I_2$ can enter the stream-initialization
branch or deposit another piece, so this ready piece remains ready until
$I_2$ reaches its piece-ready test.

It remains to consider the case in which, at the piece-transfer test of
$B_2$, $\PieceReady{i}=PA$ and $\idx{i}=L$. In this case no mailbox piece is
ready, and the writer has already deposited all $L$ pieces of the fixed
stream for $R$. Since no mailbox piece is ready, the last deposited piece has
already been acknowledged by reader $i$, meaning reader $i$ has written
$\PieceAck{i}:=\PieceReady{i}$ after reading that piece from \Mailbox{i}.
Reader $i$ acknowledges a piece only after appending it. Since reader $i$ has
no overlapping read operation, those acknowledgements after $A_R$ belong to
$R$. Therefore $R$ has already appended all $L$ pieces of the fixed stream.
By the alternating-bit delivery lemma, the sequence appended by $R$ is
exactly the whole fixed array \localVal{i}. Hence the loop condition of $R$
is already false, so $R$ cannot execute the later iteration $I_2$. This
contradicts the assumption that $I_2$ is a loop iteration of $R$.

Thus, in every possible case in which $I_2$ can occur, when $I_2$ reaches its
piece-ready test we have $\PieceReady{i}\neq\PieceAck{i}$. Since $I_2$ also
passes the first guard $\ReadAck{i}=myReq$, it appends a mailbox piece. This
contradicts the assumption that $I_2$ is unproductive. Therefore no two
consecutive loop iterations of $R$ can both be failed and unproductive.

It follows that among the failed iterations of $R$, no two unproductive
iterations occur consecutively. Hence the unproductive failed iterations can
appear only before the first productive failed iteration, after a productive
failed iteration, or not at all. Therefore the number of unproductive failed
iterations is at most one more than the number of productive failed
iterations. Thus $U \leq P+1$.
\end{proof}

\begin{lemma}[Bounded number of read iterations]
Every \ReadOp{i} operation executes at most $2L+1$ loop iterations.
\end{lemma}

\begin{proof}
Each productive failed iteration appends exactly one piece. Since the loop condition is that fewer than $L$ pieces have been appended, there can be at most $L$ productive failed iterations. Let $P$ be the number of productive failed iterations and $U$ the number of unproductive failed iterations. By the preceding lemma, $U\leq P+1$. If the read terminates through the slow path, then it terminates immediately after the $L$th productive failed iteration. Thus $P=L$, and the total number of iterations is $P+U\leq L+(L+1)=2L+1$. If the read terminates through the fast path, then it has one final successful iteration. Before that final successful iteration, it has $P\leq L-1$ productive failed iterations. Hence the total number of iterations is at most $P+U+1\leq P+(P+1)+1=2P+2\leq 2L$. In either case, the number of loop iterations is at most $2L+1$.
\end{proof}

\begin{theorem}
Every write completes in $\Theta(L+r)$ writer steps, and every read completes in
$\Theta(L^2)$ reader steps in the worst case. Therefore, the implementation is wait-free.
\end{theorem}

\begin{proof}
By the write-bound lemma, every write operation completes after $\Theta(L+r)$ writer steps.

For a read operation, each loop iteration performs a constant number of base-register accesses and one scan of a global buffer, costing $\Theta(L)$ reader steps. By the bounded-iteration lemma, a read executes at most $2L + 1$ loop iterations. 
Furthermore, a slow-path read executes at least $L$ iterations since it appends at most a single piece in each iteration. Such slow-path executions exist: after the reader announces its request, the scheduler can make the writer execute a service block for reader $i$ during each of $L$ consecutive read iterations, causing each tripwire check to fail while delivering one mailbox piece per iteration.
Therefore each read completes in $\Theta(L^2)$ reader steps in the worst case. These bounds depend only on $r$ and $L$, not on the relative speeds or failures of other processes. Therefore every operation completes in a bounded number of its own steps. The implementation is wait-free.
\end{proof}





    \section{Conclusion}
\label{sec:concl}

We study the problem of implementing $m$-valued SW registers from $b$-valued base registers of the same type and prove matching upper and lower bounds on its space complexity.
On the lower-bound side, we showed that any obstruction-free implementation with an invisible reader uses at least $\left\lceil\frac{m-1}{b-1}\right\rceil$ base registers, an exponential improvement over the trivial information-theoretic $\Omega\!\left(\frac{\log m}{\log b}\right)$ lower bound and the constant factor improvements due to Chaudhuri and Welch~\cite{chaudhuri1994bounds}. We extended this to a $\left\lceil\min\!\left(\frac{m-1}{b-1},\, r+\frac{\log m}{\log b}\right)\right\rceil$ lower bound for all implementations, whether at least one reader is invisible or every reader is visible. On the upper-bound side, we gave a wait-free implementation of a multi-word atomic register from atomic base registers using $\Theta\!\left(r+\frac{\log m}{\log b}\right)$ space. Combining this with existing $\Theta\!\left(\frac{m}{b}\right)$ space constructions~\cite{chaudhuri1994bounds,VIDYASANKAR1991323} yields an upper bound of $\Theta\!\left(\min\!\left(\frac{m}{b},\, r+\frac{\log m}{\log b}\right)\right)$, which asymptotically matches our lower bound. This improves upon the previous best upper bound of $\Theta\!\left(\min\!\left(\frac{m}{b},\, r\frac{\log m}{\log b}\right)\right)$.

These bounds are robust in two senses. The lower bound only requires the simulated register to be regular and it only requires obstruction-free progress, whereas the upper bound ensures the simulated register is atomic and guarantees wait-free progress. The two therefore pin down the space complexity for the
entire range of settings in between. 


Chen and Wei's visible-reader algorithm uses $\Theta(r\frac{\log m}{\log b})$ space and $\Theta(\frac{\log m}{\log b})$ steps per read and write~\cite{chen2017step}, whereas our asymptotically space-optimal algorithm uses $\Theta(r+\frac{\log m}{\log b})$ space, $\Theta(\frac{\log^2 m}{\log^2 b})$ read steps, and $\Theta(r+\frac{\log m}{\log b})$ write steps. Whether the optimal $\Theta(r+\frac{\log m}{\log b})$ space bound can be achieved together with $\Theta(\frac{\log m}{\log b})$ read and write step complexity remains open.
\newline\newline
\noindent\textbf{AI Disclosure.} We used Claude Opus 4.8 and GPT 5.5 to help polish writing throughout the paper. We also used Gemini 3.1 to test and generate counterexamples to earlier iterations of our algorithm. The authors verified the correctness and originality of all content, including references.

    
    \bibliographystyle{plainurl}
    \bibliography{biblio}
\end{document}